\documentclass[11pt]{article}

\usepackage[margin=1in]{geometry}
\usepackage{amsfonts,amssymb,amsmath,amsthm,mathtools}
\usepackage{thmtools,thm-restate}
\usepackage[round]{natbib}

\usepackage[T1,T2A]{fontenc}
\usepackage[utf8]{inputenc}
\usepackage[main=english,russian]{babel}

\usepackage{graphicx,latexsym,lpic,bm,xspace,booktabs,wrapfig,array}
\usepackage[protrusion=true,expansion=false]{microtype}

\usepackage[bf]{caption}

\usepackage{algorithmicx,algpseudocode}

\usepackage{dsfont}

\usepackage{enumitem}
\setlist[itemize]{noitemsep,label=$-$}
\setlist[enumerate]{noitemsep,label=\itshape(\arabic*)}

\usepackage{hyperref}
\addto\extrasenglish{}
\addto\extrasenglish{}
\newcommand{\appref}[1]{\hyperref[#1]{Appendix~\ref{#1}}}

\usepackage{algorithm}

\usepackage[usenames,dvipsnames]{xcolor}
\usepackage[framemethod=tikz]{mdframed}
\mdfsetup{
	nobreak=true,
	linewidth=.5,
	innerleftmargin=\parindent,
	innerrightmargin=20pt,
	innertopmargin=12pt,
	innerbottommargin=14pt,
	skipabove=0,skipbelow=0
}

\usepackage{tikz}
\usepackage{amsmath}
\usetikzlibrary{decorations.pathreplacing, shapes.geometric}
\usetikzlibrary{arrows.meta} 

\hypersetup{
	colorlinks=true,
	linkcolor=Sepia,
	citecolor=Sepia,
	filecolor=Sepia,
	urlcolor=Sepia
}

\usepackage{todonotes}
\presetkeys%
{todonotes}%
{inline}{}

\declaretheorem[name=Theorem]{theorem}
\declaretheorem[name=Lemma,sibling=theorem]{lemma}

\declaretheorem[name=Claim,sibling=theorem]{claim}

\declaretheorem[name=Corollary,sibling=theorem]{corollary}

\declaretheorem[name=Definition,style=definition]{definition}

\newcommand{\smallFunction}[2]{\newcommand{#1}{{\textsc{#2}}}}

\smallFunction{\Or}{Or}
\smallFunction{\Xor}{Xor}
\smallFunction{\Ind}{Ind}
\smallFunction{\IP}{IP}
\smallFunction{\MAJ}{MAJ}
\smallFunction{\Forr}{Forrelation}
\smallFunction{\Eol}{End-of-Line}

\newcommand{\newclass}[2]{\newcommand{#1}{{\text{\upshape\sffamily #2}}\xspace}}

\newclass{\BPP}{BPP}
\newclass{\RR}{R}
\newclass{\DD}{D}
\newclass{\sizeR}{sizeR}
\newclass{\sizeD}{sizeD}
\newclass{\sR}{sR}
\newclass{\sD}{sD}
\newclass{\NP}{NP}
\newclass{\RP}{RP}
\newclass{\RS}{RS}
\newclass{\ZPP}{ZPP}
\newclass{\BQP}{BQP}
\newclass{\PPAD}{PPAD}
\newclass{\WAPP}{WAPP}

\newcommand{\newsftext}[2]{\newcommand{#1}{{\text{\upshape\sffamily #2}}}}
\newsftext{\dt}{dt}
\newsftext{\cc}{cc}
\newsftext{\ccsize}{ccsize}

\newcommand{\x}{\bm{x}}

\let\OLDthebibliography\thebibliography
\renewcommand\thebibliography[1]{
  \OLDthebibliography{#1}
}

\usepackage{authblk}

\renewcommand{\Ind}{\operatorname{Ind}}

\renewcommand{\NP}{\mathrm{NP}}

\newcommand{\prop}{\mathrm{PROP}}

\newcommand{\disc}{\mathrm{disc}}

\newcommand{\oasymdiscm}{\mathrm{osasymdisc}^{\mathrm{max}}}

\author{Alexander Shekhovtsov, Georgy Sokolov, Mikhail Cherniavskii, Andrey Kupavskii}
\date{}
\title{Nearly Tight Bounds for Proportional Group Fair Divisions and One-Sided Discrepancy}

\begin{document}

\maketitle

\begin{abstract}
This paper studies the problem of fair division of indivisible goods among $k$ groups of $n_1,\ldots, n_k$ agents. We look at the worst downward deviation $\prop(n_1,\ldots, n_k)$ of an agent in a group from its $1/k$-share.  
We improve the bounds of \citep{DBLP:conf/sosa/ManurangsiM26} and show that $\prop(n_1,\ldots, n_k) = \tilde\Theta(\sqrt{n/k})$,  where $n = n_1 + \ldots + n_k$ is the total number of agents.

For the proof of the upper bound, we develop novel discrepancy-type tools and, in particular, a way to  efficiently  work with one-sided discrepancy constraints. 

\end{abstract}

\tableofcontents


\section{Introduction}
\subsection{Fair division}
Allocation of indivisible goods among several agents is an important theoretical and practical problem. In the classical setting, we have $k$ agents and a set $G$ of commodities to be divided among them. Each agent has their own nonnegative evaluation of the commodities. A typical goal is to provide a nearly-fair partition subject to different fairness notions. We refer the reader to \citep{amanatidis2023fair} for survey on fair division of indivisible goods.

In recent years, group fair divisions gained a significant amount of attention. In group fair division, individual agents are replaced by groups of agents that share the same commodities. Formally, we are given $k$ groups of $n_1\ge \ldots\ge n_k$ agents, and a set of objects $G$. Denote by $\mu^{(i,j)}(g)$  the value that the $j$-th agent in the $i$-th group accord to a commodity $g\in G$. Importantly, all valuations are nonnegative: $\mu^{(i,j)}(g)\ge 0$. The goal of the partition is again to produce nearly-fair parts, now evaluated by each agent in the group. We work with the following well-studied fairness notion. 
\begin{definition}\label{def1} Let $c$ be a positive integer. A partition $G = B_1 \sqcup B_2 \sqcup \ldots \sqcup B_k$ satisfies $\mathrm{PROP-}c$, if for any $i, j$ there is a subset $C \subseteq G \setminus B_i,~|C| \leq c$ such that 
$$
\mu^{(i, j)}(B_i) \geq \frac{\mu^{(i, j)}(G)}{k} - \mu^{(i, j)}(C).
$$
Given $k$ and $n_1,\ldots, n_k$, let us denote by
$\mathrm{PROP}(n_1, \ldots, n_k)$ the smallest $c$ such that for $k$ groups of agents of sizes $n_1,\ldots, n_k$ with any  valuations there exists a partition that satisfies  $\mathrm{PROP-}c$.
\end{definition}

In a recent paper \cite{DBLP:conf/sosa/ManurangsiM26}, the authors proved the following upper and lower bounds for this problem.
\begin{theorem}[Manurangsi--Meka, \cite{DBLP:conf/sosa/ManurangsiM26}]
We have $$\Omega\left(\max_{i\in[k]}\frac i{k}\sqrt{n_i}\right)=\mathrm{PROP}(n_1, \ldots, n_k) = O (\sqrt{n_1}).$$
\end{theorem}  
While being tight for groups of nearly equal size, these bounds are far apart for asymmetric groups. The main result of this paper is an almost tight upper and lower bounds for this setting in general.

\begin{theorem}
    \label{thm:main}
Fix some positive integers $n_1\ge \ldots \ge n_k$ and put $n = \sum_{i=1}^k n_i$. We have
    $$\Omega\left(\sqrt{\frac nk}\right)=\mathrm{PROP}(n_1, \ldots, n_k) = O \left( \sqrt{\frac{n}{k}  \cdot \left( \ln^{3/2} \left(n_1 \cdot \frac{k}{n}\right)  + 1 \right) }\right) .$$
\end{theorem}
The authors of \cite{DBLP:conf/sosa/ManurangsiM26} related the problem of determining $\mathrm{PROP}$ to a multicolor discrepancy problem. Certain tools from discrepancy theory underlie both lower and upper bounds. In earlier papers, the authors also related fair division and discrepancy questions \citep*{doerr2003multicolour,manurangsi2022almost,caragiannis2025new,DBLP:conf/sosa/ManurangsiM26}. However, all previous work always worked with two-sided discrepancy problems. In private communication, Pasin Manurangsi mentioned that developing an approach to benefit from one-sided discrepancy constraints is an important challenge. 
In this paper, we managed to develop such an approach.  

\paragraph{Our contributions}\begin{itemize}
    \item We provide nearly tight bounds for ${\rm PROP}$;
    \item We develop an approach to benefit from one-sided discrepancy constraints in discrepancy problems via the gifting mechanism;
    \item We prove a discrepancy result on the so-called Top-down Beck-Fiala problem, which could be seen as its sharper, non-uniform analogue.
\end{itemize}

Next, we describe the discrepancy theory counterpart. We  start with some background.

\subsection{Discrepancy}
Discrepancy theory studies how to color a set system so that every set remains nearly balanced. 
In the standard two-color matrix formulation, given a matrix $A \in \mathbb{R}^{n \times m}$, a {\it coloring} is a vector $\chi \in \{-1,1\}^m$, 
and the {\it discrepancy of $A$} is defined as
\[
    \disc(A) = \min_{\chi \in \{-1,1\}^m} \|A\chi\|_\infty .
\]

Several classical results provide general-purpose upper bounds for discrepancy. 
If every column of a $0$-$1$ matrix $A$ has at most $t$ nonzero entries, then $\disc(A) \leq 2t-1$~\citep*{DBLP:journals/dam/BeckF81}. 
For every matrix $A \in [-1,1]^{n \times m}$, one has $\disc(A) = O(\sqrt n)$~\citep*{spencer1985six}. 
More generally, for vectors of small Euclidean norm and any convex body of Gaussian measure at least $1/2$, there are signs whose signed sum lies in that body; in particular, this implies an $O(\sqrt{\log n})$ discrepancy bound in the $\ell_\infty$ norm for matrices with $n$ rows and columns of unit $\ell_2$ norm~\citep*{DBLP:journals/rsa/Banaszczyk98}. 
Recent work improves Banaszczyk-type bounds for the Beck--Fiala and Koml{\'o}s settings and gives efficient algorithms for finding the corresponding colorings~\citep*{DBLP:conf/stoc/BansalJ26}. 

For the fair-division setting considered here, the most relevant notion is that of multi-color and multi-matrix  discrepancy.

\begin{definition}[$k$-matrix discrepancy]
Let $A^1 \in [0, 1]^{n_1 \times m}, \ldots, A^k \in [0, 1]^{n_k \times m}$ be matrices. Define $\disc^k(A^1, \ldots, A^k)$ to be the smallest $c$, such that there exists a partition $v_1, \ldots, v_k \in \{0, 1\}^m,$ $\sum v_i = \bm{1},$ with  the following property. For all $i \in [k]$ and $ j \in [n_i]$ we have
$$
\langle A^i_j, \bm{1}/k\rangle - c\le \langle A^i_j, v_i\rangle \leq \langle A^i_j, \bm{1}/k\rangle + c.
$$
    Here $A^i_j$ denotes the $j$-th row of matrix $A^i$.
\end{definition}
$k$-matrix discrepancy was called asymmetric $k$-color discrepancy in \cite{DBLP:conf/sosa/ManurangsiM26}. We put $$\disc^k(n_1,\ldots,n_k) = \max \disc(A^1,\ldots, A^k),$$
where the maximum is taken over $A^i\in [0,1]^{n_i\times m}$, where $m$ could be arbitrary. 
We use the same notational conventions for other types of discrepancy.

\begin{definition}[One-sided $k$-matrix discrepancy]
Let $A^1 \in [0, 1]^{n_1 \times m}, \ldots, A^k \in [0, 1]^{n_k \times m}$ be matrices. Define $\disc^k_>(A^1, \ldots, A^k)$ similar to $\disc^k$, only with a different set of inequalities imposed: 
$$
\langle A^i_j, v_i\rangle \geq \langle A^i_j, \bm{1}/k\rangle - c
$$
for each row $A^i_j$ of $A^i$.
\end{definition}

Note that we may renormalize the valuations of each agent so that for each $i,j$ we have 
$$\max_{g\in G} \mu^{(i,j)}(g)=1.$$ If there is a ${\rm PROP-}c$ partition when agents utilities are represented by matrices $A^1,\ldots, A^k$, then $\disc^k_>(A^1,\ldots, A^k)\le c$.  It implies the left inequality in the theorem below.
\begin{lemma}
\label{thm:osasprop}
For any $k \geq 2$ and $n_1 \geq \ldots \geq n_k \geq 1$,
$$
\disc^k_>(n_1, \ldots, n_k) \leq \prop(n_1, \ldots, n_k) \leq 2 \cdot \lceil \disc^k_>(n_1, \ldots, n_k) \rceil
$$
\end{lemma}
The right inequality  was essentially proved in  \citep[Theorem 5.3]{DBLP:conf/sosa/ManurangsiM26}.
They have also showed the following.
    $$
    \disc^k(n_1, \ldots, n_k) = \Theta(\sqrt{n_1}).
    $$
Using Lemma~\ref{thm:osasprop}, we get that Theorem~\ref{thm:main} is equivalent to the following result.
\begin{theorem}
    \label{thm:main2}
Fix some positive integers $n_1\ge \ldots \ge n_k$ and put $n = \sum_{i=1}^k n_i$. We have
    $$\Omega\left(\sqrt{\frac nk}\right)=\disc^k_>(n_1, \ldots, n_k) = O \left( \sqrt{\frac{n}{k}  \cdot \left( \ln^{3/2} \left(n_1 \cdot \frac{k}{n}\right)  + 1 \right) }\right) .$$
\end{theorem}
The bounds in  Theorem~\ref{thm:main2} and the displayed equality before it are very different for skewed matrix sizes. Thus the one-sided nature of the problem makes it genuinely different.  Little was known concerning  rounding of one-sided constraints from this point of view. In this work, we  develop the techniques that take advantage of the one-sided inequalities. Specifically, we introduce  the gifting algorithm, described in Section~\ref{sec211}. As is witnessed by our main result, the resulting approach is optimal up to polylog factors. Below, we summarize the main ideas of the proofs.

\paragraph{Our techniques.}
The proof of the upper bound has two main ingredients. The first ingredient is Theorem~\ref{thm:delta-struct-division} that essentially allows to bound $\disc^k$ using a binary tree-type procedure. In its core is a discrepancy theory result (Top-down Beck-Fiala) which we believe is interesting in its own right. In a nutshell, it is based on a careful accounting of the $\ell_1$-norms of rows and uses the Partial Coloring Lemma. It allows to get discrepancy bounds that incorporate average row complexity after several columns were deleted.

The problem with Theorem~\ref{thm:delta-struct-division} is that it gives worse errors for larger groups and for rows with larger $\ell_1$-norm. As a result, its direct application would lead to suboptimal bounds. To address this issue, we introduce a rebalancing procedure called the {\em gifting algorithm}, which carefully distributes a small part of commodities between (mostly large) groups. As a result, the most problematic agents (from large groups and with large $\ell_1$-norm of their valuation get the most gifts. This provides us with a necessary correction that then allows to apply Theorem~\ref{thm:delta-struct-division}. We note that the absence of the inequality from the other side spares us from a more elaborate tracking the weight of commodities that these `heavy' agents get and gives us the necessary freedom in action.

Our lower bound is inspired by the construction used in \cite{DBLP:conf/sosa/ManurangsiM26}. The improvement comes from the intuition that rows of very different $\ell_1$-norm should make fair division more complicated. In the construction, we have some groups consisting of all-ones vectors, and the other (larger) groups get sparse valuations that look like block-diagonal matrices with the construction of Manurangsi and Meka as a block.

\paragraph{Organization.} In Section~\ref{secBeckFiala}, we formulate our generalization of the Beck--Fiala problem. In Section~\ref{XiBound}, we prove an upper bound for this problem. In Section~\ref{sec_upper_bound}, we use this result to obtain the upper bound in Theorem~\ref{thm:main}. The gifting algorithm, another important part of the proof of the upper bound, is presented in Section~\ref{sec211}. Finally, in Section~\ref{sec_lower_bound}, we prove the lower bound in Theorem~\ref{thm:main}.

\subsection{Discrepancy and the Top-down Beck-Fiala Problem} \label{secBeckFiala}
In the last decade, there has been great progress on the algorithmic aspects of discrepancy. One major development was the paper  \cite{DBLP:conf/focs/LovettM12}, where the authors provided a polynomial algorithm for the classical Spencer's `six standard deviations' result \cite{spencer1985six}. They introduced the following  powerful tool. 

\begin{theorem}[Partial Coloring Lemma, \citep{DBLP:conf/focs/LovettM12}]
\label{thm:main-partial-coloring}
Let $v_1,\dots,v_n \in \mathbb{R}^m$ be vectors, and $y \in [-1,1]^m$ be a ``starting'' point. Let $c_1,\dots,c_n \geq 0$ be thresholds such that $\sum_{i=1}^n \exp(-c_i^2/16) \leq m/16$. Let $\delta > 0$ be a small approximation parameter. Then there exists an efficient randomized algorithm that with probability at least $0.1$ finds a point $y' \in [-1,1]^m$ such that

\begin{itemize}
    \item[(i)] $|\langle y - y', v_i \rangle| \leq c_i \|v_i\|_2$.
    \item[(ii)] $|y'_i| \geq 1 - \delta$ for at least $m/2$ indices $i \in [m]$.
\end{itemize}

Moreover, the algorithm runs in time $O((m+n)^3 \cdot \delta^{-2} \cdot \log(nm/\delta))$.
\end{theorem}

Another classical discrepancy result is the Beck--Fiala theorem \citep{DBLP:journals/dam/BeckF81}, which states that, if each column of $A$ has $\ell_1$-norm at most $\Delta$, then the upper bound can be improved to $O(\Delta)$. They conjectured that the right upper bound is $O(\sqrt{\Delta})$. \cite{DBLP:journals/rsa/Banaszczyk98} proved a $O(\sqrt{\Delta \cdot \ln n})$ bound, and the result of \cite{DBLP:conf/focs/LovettM12} gives the same bound and a polynomial construction. In a recent breakthrough paper,  \cite{DBLP:conf/stoc/BansalJ26} showed the $O(\sqrt{\Delta})$ bound in the regime $\Delta \geq \log^2 n$. Still, the question remains open in general. 


In this paper, we prove and apply the following theorem, which deals with an asymmetric generalization of the Beck-Fiala problem.\footnote{We state it for $[0,1]$ for convenience. The result also works for $[-1, 1]$  as well, with column sums of absolute values instead of column sums.}

\begin{definition}[$(\Delta_i)$-structured matrix]
Given positive reals $\Delta_1 \geq \Delta_2 \geq \ldots \geq \Delta_n$, we say that a matrix $A \in [0,1]^{n \times m}$ is {\em $(\Delta_i)$-structured} if for any $i\in [n]$, the maximum column sum in a submatrix $A_{i, \ldots, n}$ is at most $\Delta_i$. Formally,
$$
\forall~ 1 \leq i \leq n, ~~\max_{j\in[m]} \sum_{k=i}^n A_{k, j} \leq \Delta_i.
$$
\end{definition}

We remark that  the setting of the Beck--Fiala theorem and, importantly, the procedure underlying its proof, works with the setting  $\Delta_1 = \ldots = \Delta_n$. However, as we illustrate in our proofs, one could draw benefits from this asymmetric setting.  Let us define the linear discrepancy problem for $(\Delta_i)$-structured matrices.




\begin{theorem}[Bound on the top-down Beck-Fiala problem]
\label{thm:logbound}
    There exists an absolute constant $C > 0$, such that the following holds. Let $n \geq \Delta_1 \geq \Delta_2 \geq \ldots \geq \Delta_n > 0$ be a sequence of reals.  Let $A \in [0, 1]^{n \times m}$ be a $(\Delta_i)$-structured matrix and $x_0 \in [0, 1]^m$ be a point. Then there exists $x \in \{0, 1\}^m$, such that for all $1 \leq i \leq n$ we have
    $$
     |\langle A_i, x_0 - x \rangle| \leq C \log_2 (4n / \Delta_n) \cdot \sqrt{\Delta_i\log_2 (4\Delta_1 / \Delta_n)}.
    $$
\end{theorem}

We prove this theorem in Section \ref{XiBound}. For convenience,  let us introduce the following function
\begin{align} \label{eq_xi_def}
    \Xi(s, r) := C \log_2(4s) \cdot \sqrt{\log_2 (4r)},
\end{align}
where $C$ is a constant from Theorem \ref{thm:logbound}. With this notation, the bound from Theorem \ref{thm:logbound} becomes 
$$|\langle A_i, x_0 - x \rangle| \leq\Xi(n/\Delta_n, \Delta_1/\Delta_n) \cdot \sqrt{\Delta_i}.$$
The proof is based on the idea that we can bound the number of rows having $\ell_1$-norm above a certain threshold. The idea of bounding large-norm rows has been instrumental in discrepancy. It has been used in random walk techniques \citep{DBLP:conf/focs/LovettM12, DBLP:conf/stoc/BansalJ26} and in the foundational work of \citep{DBLP:journals/dam/BeckF81}. Besides this, we use the partial coloring lemma from \citep{DBLP:conf/focs/LovettM12} and draw on some of the ideas from that paper.
We suspect that 
the polylogarithmic term $\sqrt{\log_2(4r)}$ in Theorem \ref{thm:logbound} is  necessary, see Lemma~\ref{lem:counter1}. 
\subsection{Related work}
\paragraph{Group fair division.}
Fair division among groups (families), models settings in which each group receives a common bundle while agents within the same group may have different valuations. 
For divisible resources, several fairness notions for families have been studied, including average, unanimous, and democratic variants of proportionality and envy-freeness~\citep*{segal2019fair}. 
Another line of work considers cake divisions among groups of prescribed sizes, with the additional requirement that every group receives a connected piece~\citep*{segal2021how}.

For indivisible goods, early work studied the asymptotic existence of envy-free allocations for groups under random additive valuations~\citep*{manurangsi2017asymptotic}, as well as approximate maximin-share guarantees for groups~\citep*{suksompong2018approximate}. 
Since unanimous guarantees can be too demanding when agents in the same group have heterogeneous preferences, democratic fairness notions have also been proposed for indivisible goods~\citep*{segal2019democratic}. 
Further work studied EF1- and EFX-type guarantees for group allocations under several valuation classes and group models~\citep*{kyropoulou2020almost}. 
More recent directions include ordinal maximin-share guarantees, variable groups, and allocations among couples or other small groups~\citep*{manurangsi2025ordinal,golz2026fair,dupre2026bad}.

The line of work closest to the present paper connects group fair division with discrepancy theory. 
There, the properties of approximate proportionality, envy-freeness, and consensus division for groups are translated into multi-color and weighted discrepancy problems~\citep*{manurangsi2022almost}. 
Subsequent sharper lower bounds used further developed the apparatus of  multi-color and weighted discrepancy techniques~\citep*{caragiannis2025new,DBLP:conf/sosa/ManurangsiM26,dupre2026bad}.

\paragraph{Discrepancy} Constructive discrepancy results often rely on random walk and semidefinite programming ideas~\citep*{bansal2010constructive,DBLP:conf/focs/LovettM12,rothvoss2017constructive,bansal2019algorithm,bansal2018gram,DBLP:conf/stoc/BansalJ26}. A more detailed introduction to discrepancy theory can be found in the standard books on the discrepancy method and geometric discrepancy~\citep*{chazelle2000discrepancy,matousek1999geometric}.

In the fair-division setting considered here, the relevant notions are not only ordinary two-color discrepancy, but also multi-color and weighted variants. 
For $k$ colors, one natural objective is to minimize, over all relevant rows, the largest imbalance between any two color classes. 
Multi-color discrepancy and its weighted variants provide a natural language for studying allocations among several groups: colors can be interpreted as groups, and a coloring can be viewed as rounding the ideal fractional allocation in which every group receives a $1/k$ fraction of every good~\citep*{doerr2002discrepancy,doerr2003multicolour}. 
This viewpoint underlies discrepancy-based bounds for group fair division, including bounds for approximate proportionality, envy-freeness, and consensus division~\citep*{manurangsi2022almost,caragiannis2025new,DBLP:conf/sosa/ManurangsiM26,dupre2026bad}. 
Recent work also extends discrepancy-based fair-division techniques beyond additive valuations by defining discrepancy notions for non-additive valuations and applying them to fair-division problems~\citep*{la2025discrepancy}.

\subsection{Notation}
We denote $[n] = \{1, 2, \ldots, n\}$.

For a matrix $A \in \mathbb{R}^{n \times m}$ and an index $i \in [n]$, by $A_i \in \mathbb{R}^m$ we denote the vector representing the $i$-th row of the matrix $A$. If $I \subseteq [n]$, then we denote by $A_I \in \mathbb{R}^{|I| \times m}$ the submatrix of $A$ obtained by retaining only the rows $i$ that lie in $I$. If $1 \leq i \leq j \leq n$, then we define $A_{i \ldots j}$ to be $A_I$ where $I = \{i, i + 1, \ldots, j\}$. For a vector $v \in \mathbb{R}^m$ and a subset $I \subseteq [m]$, we denote by $v_I \in \mathbb{R}^{|I|}$ the restriction of $v$ to the coordinates $I$.

For a vector $v \in \mathbb{R}^m$, let $\| v\|_1, \|v\|_2, \|v\|_{\infty}$ denote its $1$-norm, $2$-norm, and $\infty$-norm, respectively.

\section{Proof of Theorem~\ref{thm:logbound}}
\label{XiBound}

We first prove the following simple, but important lemma.

\begin{lemma}
    \label{lem:rowbound}
    Let $n, m \in \mathbb{N}$ and $A \in [0, 1]^{n \times m}$ be a $(\Delta_i)$-structured matrix where $\Delta_1/\Delta_n \leq r$ for $r > 0$. Then, for any $t > 0$, the number of rows $i$ which satisfy
    \begin{equation}
     \|A_i\|_1 > 2 \cdot t \cdot (\log_2 r + 2) \cdot \Delta_i \label{eqn:exceed}
    \end{equation}
    is at most $m/t$.
\end{lemma}
\begin{proof}
    Let $k$ be an integer such that $(\log_2 \Delta_n) - 1\leq k \leq \log_2 \Delta_1$ and assume there is at least one $i$, such that $2^{k} \leq \Delta_i \leq 2^{k+1}$. Let $d$ be the smallest index, such that $\Delta_d \leq 2^{k+1}$, and let $u \geq d$ be the largest index, such that $\Delta_u \geq 2^{k}$.

    The sum in each column of the matrix $A_{d, \ldots, u}$ is at most $2^{k+1}$. So, the total sum of all entries of rows $d, \ldots, u$ is at most $m \cdot 2^{k+1}$. Therefore, the number of rows $i \in [d, u]$ that have sum of entries exceeding $2 \cdot t \cdot (\log_2 r + 2) \cdot \Delta_i $ is at most
    $$
    \frac{m \cdot 2^{k+1}}{2 \cdot t \cdot (\log_2 r + 2) \cdot 2^{k}} = \frac{m}{t \cdot (\log_2 r + 2)},
    $$
    where we used $\Delta_i \geq 2^k$.

    Since for every $\Delta_i$, there is an integer $(\log_2 \Delta_n) - 1 \leq k \leq \log_2 \Delta_1$, such that $2^k \leq \Delta_i \leq 2^{k+1}$, summing over all such $k$, we obtain that the number of rows satisfying (\ref{eqn:exceed}) is at most

    $$
    (\log_2 r + 2) \cdot \frac{m}{t \cdot (\log_2 r + 2)} = \frac{m}{t}.
    $$

\end{proof}

We remark that the $\log_2 r$ term in Lemma \ref{lem:rowbound} cannot be reduced as shown by the following lemma.

\begin{lemma}
\label{lem:counter1}
There exists an absolute constant $m_0 > 0$, such that for every $m \geq m_0$, there exists $n \geq 2 \cdot m$, positive integers $(\Delta_i)$ where $n \geq\Delta_1 \geq \Delta_2 \geq \ldots \geq \Delta_n \geq 1$, and a matrix $A \in [0, 1]^{n \times m}$ that is $(\Delta_i)$-structured and all its rows satisfy
$$
     \|A_i\|_1 > \frac{\ln m}{3} \cdot \Delta_i.
    $$
\end{lemma}
\begin{proof}
    We construct the matrix $A$ as a vertical concatenation of matrices $A_1, \ldots, A_k$, where the maximum column sum in each matrix $A_i$ is $1$:
    $$
A = \begin{bmatrix}
A_{k} \\[1ex]
A_{k-1} \\[1ex]
\vdots \\[1ex]
A_1
\end{bmatrix}.
$$
If we look at a row $r$ of matrix $A$ that is a part of a matrix $A_i$, then the maximum column sum in rows $\geq r$ is at most $i$. Let us set $\Delta_r = i$.

For each $i$, we design matrix $A_i$ to have $(i \cdot \ln m)/3$ ones in each row. We can make $A_i$ to have at least $\lfloor 3m/(i \cdot \ln m) \rfloor$ rows, so that ones in each row are in disjoint set of columns. Therefore, choosing $k=\lfloor m / \ln m \rfloor$ the total number of rows $n$ across all matrices is

$$
n \geq \sum_{i=1}^{k} \lfloor 3m/(i \cdot \ln m) \rfloor \geq \sum_{i=1}^{k} \left(3m/(i \cdot \ln m) - 1 \right) = m/\ln m \cdot 3\sum_{i=1}^{\lfloor m/ \ln m \rfloor} \frac{1}{i} - m/\ln m \geq 2 \cdot m
$$
for sufficiently large $m \geq m_0$.
\end{proof}

We rely on the following corollary from \citep{DBLP:conf/focs/LovettM12}.

\begin{corollary}[corollary from \citep{DBLP:conf/focs/LovettM12}]
\label{cor:lovett}
Let $A \in [0, 1]^{n \times m}$ be a matrix and let $c_1, \ldots, c_n \geq 0$ be such that $\sum_{j=1}^n \exp(-c_j^2/16) \leq m/16$. Let $\delta > 0$ be a small approximation parameter. Then there exists an efficient randomized algorithm which for any point $x \in [0, 1]^m$ with probability at least $0.1$ finds a point $x' \in [0,1]^m$ such that

\begin{itemize}
    \item[(i)] $|\langle x - x', A_i \rangle| \leq c_i \|A_i\|_2$.
    \item[(ii)] $x'_i \geq 1 - \delta$ or $x'_i \leq \delta$ for at least $m/2$ indices $i \in [m]$.
\end{itemize}

Moreover, the algorithm runs in time $\textit{poly}(n, m, \delta^{-1})$.
\end{corollary}

We also need to use the following standard lemma.

\begin{lemma}[\cite{DBLP:journals/dam/BeckF81}]
    \label{lem:kerslide}
    Let $m > n > 0$ be integers, let $A \in [0, 1]^{n \times m}$ be a matrix, and $x \in [0, 1]^m$ be a point. Then, there exists a point $x' \in [0, 1]^m$, such that $A x = Ax'$ and all but $n$ coordinates of $x'$ are rounded (either $1$ or $0$).
\end{lemma}


Now, let us recall the statement of Theorem \ref{thm:logbound} and pose its proof. Our proof is based on the ideas that \citep{DBLP:conf/focs/LovettM12} use to prove their Theorem $2$.

\begin{theorem}[Restatement of Theorem \ref{thm:logbound}]
    There exists an absolute constant $C > 0$, such that the following holds. Let $n \geq \Delta_1 \geq \Delta_2 \geq \ldots \geq \Delta_n > 0$ be a sequence of reals.  Let $A \in [0, 1]^{n \times m}$ be a $(\Delta_i)$-structured matrix and $x_0 \in [0, 1]^m$ be a point, then there exists $x \in \{0, 1\}^m$, such that for all $1 \leq i \leq n$,
    $$
     |\langle A_i, x_0 - x \rangle| \leq C \log_2 (4n / \Delta_n) \cdot \sqrt{\Delta_i\log_2 (4\Delta_1 / \Delta_n)}.
    $$
\end{theorem}

\begin{proof}[Proof of Theorem \ref{thm:logbound}]
Let $r = \Delta_1 / \Delta_n$. By Lemma \ref{lem:kerslide} we can round $m - n$ coordinates of $x_0$ if $m > n$. Therefore, let us assume that $m \leq n$.

Let us choose $c_i = C \cdot \sqrt{(\log_2 (r) + 2) \cdot \Delta_i}/{\|A_i\|_2}$ and show that $\sum_{i=1}^n \exp(-c_i^2/16) \leq m/16$ for $C$ to be chosen later. For each integer $s$, let $q_s$ denote the number of rows for which it holds

$$
2^s \cdot (\log_2(r) + 2) \cdot \Delta_i \leq ||A_i||_2^2 \leq 2^{s+1} \cdot (\log_2(r) + 2) \cdot \Delta_i.
$$
Note that $\|A_i\|_2^2 \leq \|A_i\|_1$ and, thus, by Lemma \ref{lem:rowbound}, we have $q_s \leq m/2^{s-1}$. Thus,
$$\sum_{i=1}^n \exp(-c_i^2/16) \leq \sum_{s=-\infty}^{\infty} q_s \cdot \exp\left(-\frac{C^2}{16 \cdot 2^{s+1}}\right) \leq$$
$$\leq m \cdot \sum_{s=-\infty}^{\infty} \exp\left(-\frac{C^2}{16 \cdot 2^{s+1}}\right)/2^{s-1} = $$
$$= m \cdot 2 \cdot \sum_{s=-\infty}^{\infty} 2^s \cdot \exp\left(-2^s \cdot \frac{C^2}{32}\right) \leq $$
$$\leq m \cdot 2 \cdot \left(\sum_{s=0}^{\infty} 2^s \cdot \exp\left(-2^s \cdot \frac{C^2}{32}\right) + \sum_{s=0}^{\infty} 2^{-s} \cdot \exp\left(- \frac{C^2}{32 \cdot 2^s}\right)\right) \leq m/16$$
for sufficiently large $C > 0$.

Now, apply Corollary \ref{cor:lovett} for $c_1, \ldots, c_n$ and $\delta = \sqrt{\Delta_n}/(3n^2)$, obtaining a vector $x_1$ that has at least $m/2$ coordinates nearly rounded and $|\langle x_1 - x_0, A_i \rangle| \leq c_i \|A_i\|_2 = C \cdot \sqrt{(\log_2 (r) + 2) \cdot \Delta_i}$. Let $I_1 \subseteq [m]$ denote the coordinates of $x_1$ that are not $\delta$-close to $0$ or $1$. Let us recursively apply this procedure to $(x_1)|_{I_1}$, obtaining a new point $x_2 \in [0, 1]^m$ with a set $I_2 \subseteq I_1$ of at most $m/4$ coordinates not $\delta$-close to $0$ or $1$. By iterating this at most $\log_2 (m/\Delta_n) + 2$ times, we get a sequence of points $x_1, x_2, \ldots, x_l$, for which $l \leq \log_2 (m/ \Delta_n) + 2$ and $x_l$ has all but $\min(\Delta_n, n/4)$ coordinates nearly rounded, denoted by $I_l$. By triangle inequality, for each $i$ we have
$$
|\langle x_l - x_0, A_i \rangle| \leq C \cdot (\log_2 (m / \Delta_n) + 2) \cdot \sqrt{(\log_2 (r) + 2) \cdot \Delta_i}.
$$
Now, let us nearly round the remaining $m' \leq \min(\Delta_n, n/4)$ coordinates. For this, we choose $c_i = 8 \cdot \sqrt{\ln(n/m')}$. Then, $\sum_{i=1}^n \exp(-c_i^2/16) = \sum_{i=1}^n \exp(-4 \cdot \ln(n/m')) =  n \cdot (m'/n)^4 \leq m'/16$. Applying Corollary \ref{cor:lovett} to the sequence $c_i$, $\delta=\sqrt{\Delta_n}/(3n^2)$, and point $(x_l)|_{I_l}$, we get a point $x_{l + 1} \in [0, 1]^m$ that has at most $m'/2$ coordinates $I_{l+1} \subseteq I_l$ not $\delta$-close to $0$ or $1$ and $|\langle x_l - x_{l+1}, A_i\rangle| \leq c_i \cdot \sqrt{m'} = 8 \sqrt{\ln(n/m')} \cdot \sqrt{m'}$ for every $i$. We iterate this procedure until all coordinates get nearly rounded. Denote $x_{l+l'}$ to be the resulting point. Then, by triangle inequality, for every $i$,

$$
|\langle x_l - x_{l+l'}, A_i\rangle| \leq 8 \cdot \sum_{i=0}^{\infty} \sqrt{\ln(2^i \cdot n/m')} \cdot \sqrt{m' / 2^i} \leq C' \cdot \sqrt{\ln(n/m') \cdot m'}
$$
for some constant $C' > 0$.

Let $x$ be $x_{l+l'}$ in which we round each coordinate to the nearest integer. Then, the error of the $i$-th row is estimated by

$$
|\langle x - x_0, A_i\rangle| \leq C \cdot (\log_2 (m / \Delta_n) + 2) \cdot \sqrt{(\log_2 (r) + 2) \cdot \Delta_i} + C' \cdot \sqrt{\ln(n/m') \cdot m'} + \frac{\sqrt{\Delta_n}}{3n} \leq
$$
$$
\leq C \cdot (\log_2 (m / \Delta_n) + 2) \cdot \sqrt{(\log_2 (r) + 2) \cdot \Delta_i} + C' \cdot \sqrt{\ln(n/\Delta_n) \cdot \Delta_n} + \frac{\sqrt{\Delta_n}}{3n} \leq
$$
$$
\leq (C + C' + 1) \cdot (\log_2 (n / \Delta_n) + 2) \cdot \sqrt{(\log_2 (r) + 2) \cdot \Delta_i} =
$$
$$
(C + C' + 1) \cdot \log_2 (4n / \Delta_n) \cdot \sqrt{\Delta_i\log_2 (4r)}.
$$
\end{proof}

In fact, in the proof of the upper bound for $\mathrm{PROP}$, we shall need a slight generalization of Theorem~\ref{thm:logbound} for a vertical concatenation of $(\Delta_i)$-structured matrices.

\begin{theorem}
\label{th:comb1}
    Let $n_1, \ldots, n_k$ and $m$ be positive integers. Let $A^1 \in [0, 1]^{n_1 \times m}, \ldots, A^k \in [0, 1]^{n_k \times m}$ be $(\Delta_j^1), (\Delta_j^2), \ldots, (\Delta_j^k)$-structured matrices respectively. Let $r$ be the maximum ratio $r = \max_{i, j, i', j'} \frac{\Delta^{i}_j}{\Delta^{i'}_{j'}}$, and let $d$ be the minimum $d = \min_{i, j} \Delta^i_j$. Then, for any $x_0 \in [0, 1]^m$, there exists $x \in \{0, 1\}^m$, such that for all $i=1 \ldots k$ and $j= 1 \ldots n_i$
    $$
    |\langle A_j^i, x_0-x \rangle| \leq \Xi((n_1 + \ldots + n_k) / d, r) \cdot \sqrt{k \cdot \Delta_j^i}.
    $$
\end{theorem}

\begin{proof}
    Let $n = n_1 + \ldots + n_k$ denote the total number of rows in all matrices. Let us sort all values of $\Delta^i_j$ in decreasing order: $\Delta_{j_1}^{i_1} \geq \Delta_{j_2}^{i_2} \geq \ldots \geq \Delta_{j_n}^{i_n}$. Consider a matrix $A$ obtained from matrices $A^1, \ldots, A^k$ by stacking vertically their rows in the order $(i_1, j_1), (i_2, j_2), \ldots, (i_n, j_n)$
   $$
A = \begin{bmatrix}
A_{j_1}^{i_1} \\[1ex]
A_{j_2}^{i_2} \\[1ex]
\vdots \\[1ex]
A_{j_n}^{i_n}
\end{bmatrix}
$$

We claim that this is a $(k \cdot \Delta_{j_l}^{i_l})_{l=1}^n$-structured matrix. Indeed, consider $l$-th row of this matrix. Since we sorted $\Delta_{j}^{i}$ in decreasing order, only rows that have $\Delta$ less or equal to $\Delta_{j_l}^{i_l}$ are below $l$-th row. Let us fix one of the initial matrices $r \in \{1, 2, \ldots, k\}$ and look only at the subset of its rows $I = \{q ~|~\Delta_{q}^r \leq \Delta_{j_l}^{i_l}\}$. By definition, we know that the maximum sum in a column in $A^r_I$ is at most $\Delta_{j_l}^{i_l}$. Therefore, the maximum sum in a column of the submatrix
$$
\begin{bmatrix}
A_{j_l}^{i_l} \\[1ex]
A_{j_{l+1}}^{i_{l+1}} \\[1ex]
\vdots \\[1ex]
A_{j_n}^{i_n}
\end{bmatrix}
$$

can be bounded by $k \cdot \Delta_{j_l}^{i_l}$ by summing over all $k$ matrices. Now, applying Theorem \ref{thm:logbound} to matrix $A$, we obtain the desired bound of $\Xi(n/d, r) \cdot \sqrt{k \cdot \Delta_{j_l}^{i_l}}$ on linear discrepancy.
\end{proof}

\section{Upper Bound in Theorem~\ref{thm:main2}} \label{sec_upper_bound}
Recall, that the goal is to prove
    $$\disc^k_>(n_1, \ldots, n_k) = O \left( \sqrt{\frac{n}{k}  \cdot \left( \ln^{3/2} \left(n_1 \cdot \frac{k}{n}\right)  + 1 \right) }\right) .$$

Throughout this section, we  use group fair division interpretation of $\disc(A^1, \ldots, A^k)$. Let $n_1 \geq \ldots \geq n_k \geq 1$ be the group sizes, where $n = n_1 + \ldots + n_k$ is the total number of agents. The matrices $A^1 \in [0, 1]^{n_1 \times m}, \ldots, A^k \in [0, 1]^{n_k \times m}$ denote agents utilities, and the goal is to find a partition $v_1, \ldots, v_k \in \{0, 1\}^m,  ~\sum v_i = \bm{1}$, such that for any $i, j$ we have $$\langle A^i_j,  v_i\rangle \geq \langle A^i_j, \bm{1}/k \rangle - c,$$ with $c$ as small as possible.

The following theorem describes a partition algorithm that control the discrepancy of rows in terms of the $(\Delta_j)$-structure of the group matrices.

\begin{theorem}
    \label{thm:delta-struct-division}
    There exists an absolute constant $C' > 0$, such that the following holds. Suppose agents valuations are presented by matrices $A^1 \in [0, 1]^{n_1 \times m}, \ldots, A^k \in [0, 1]^{n_k \times m}$ that are $(\Delta_j^1), (\Delta_j^2), \ldots, (\Delta_j^k)$-structured, respectively. Let $r$ upper bound the ratio between $\Delta^i_j$: $$\max_{i, j, i', j'} \frac{\Delta^{i}_{j}}{\Delta^{i'}_{j'}} \leq r.$$ 
    Let $d \leq n_1$ be the lower bound on the minimum value of $\Delta:$ $d \leq \min_{i, j} \Delta^i_j$. Then, there is an allocation $v_1, \ldots, v_k \in \{0, 1\}^m$, $\sum v_i = \bm{1}$, such that for any $i \in [k]$ and $j \in [n_i]$ we have
    $$
    |\langle A^i_j, v_i \rangle - \langle A^i_j,\bm{1} / k \rangle| \leq C' \cdot \Xi(n_1/d, r) \cdot \sqrt{\Delta^i_j}
    $$
    \end{theorem}
\begin{proof}
The role of this theorem is to provide a high-level partition of our goods. 
Let us recursively divide objects between groups. If $k = 1$, we give all the objects to the only group, incurring zero error.  For $k \geq 2$, we divide all groups into two parts. The first part consists of the first $\lfloor k/2 \rfloor$ groups, and the second part consists of the remaining $\lceil k/2 \rceil$ groups. We apply Theorem \ref{th:comb1} with $x_0 = \bm{1} \cdot \frac{\lfloor k/2 \rfloor}{k}$ to obtain a rounded vector $x \in \{0, 1\}^m$. We recursively divide objects that correspond to ones in $x$ between the first $\lfloor k/2 \rfloor$ groups and divide objects that correspond to zeros in $x$ between the remaining $\lceil k/2 \rceil$ groups. Let $v_i$ be the indicator vector of the set of objects that the $i$-th groups receives at the end of this division.

This division procedure can be viewed as a binary tree. In every node of the tree, we have a segment of groups between which we divide the set of objects assigned to this node. If this is a non-leaf node, then some objects go to the left child and then rest of the objects go to the right child.

For the $j$-th agent from the $i$-th group, let us estimate his evaluation of the goods received. Take the path $P$ from the root of the binary tree to the leaf which corresponds to the $i$-th group. Define $k_t$ to be the number of groups that participated in the division in $t$-th node of path $P$. We have $k_1 = k$ and $k_l = 1$ where $l$ is the number of nodes in $P$. Moreover, we also have $k_{t+1}$ to be equal either to $\lfloor k_t/2\rfloor$ or $\lceil k_t/2 \rceil$ for every $1 \leq t \leq l - 1$.

For $t=1 \ldots l$, let $x_t \in \{0, 1\}^m$ denote the indicator of the set of goods that was divided in the $t$-th node of the path $P$. We have $x_1 = \bm{1}$ and $x_l = v_i$. Let us check that

\begin{equation}
\label{eq:nodebound}
\left| \langle A^i_j, x_{t+1} \rangle - \langle A_j^i,  x_t \frac{k_{t+1}}{k_t} \rangle \right| \leq \Xi(n_1 \cdot k_t/d, r) \cdot \sqrt{k_t \cdot \Delta_j^i}.
\end{equation}

To divide objects in the $t$-th node, we apply Theorem \ref{th:comb1} to the submatrix of $A$ narrowed down to rows and columns that correspond to agents and objects that participate in this division. The theorem gives us a vector $x''$. Let $x' \in \{0, 1\}^m$ be the vector $x''$ padded with zeros for objects that did not participate in this division.  By guarantees of Theorem \ref{th:comb1}, $x'$ satisfies

$$
\left| \langle A^i_j, x' \rangle - \langle A_j^i,  x_t \frac{\lfloor k_t/2 \rfloor}{k_t} \rangle \right| \leq \Xi(n_1 \cdot k_t/d, r) \cdot \sqrt{k_t \cdot \Delta_j^i}.
$$

Here we used that the total number of agents in groups that participated in the division in the $t$-th node of the path $P$ can be bounded by $n_1 \cdot k_t$, since $n_1$ is the largest number of agents in any group.

If the $i$-th group is in the first $\lfloor k_t/2\rfloor$ groups in this division, then we have $x_{t+1} = x'$, concluding (\ref{eq:nodebound}) in this case. If it is among the remaining $\lceil k_t / 2 \rceil$ groups, then we have $x_{t+1} = x_t - x'$. Since $1 =  \frac{\lfloor k_t / 2\rfloor + \lceil k_t / 2\rceil }{k_t}$, we have

$$
\left| \langle A^i_j, x_{t+1} \rangle - \langle A_j^i,  x_t \frac{k_{t+1}}{k_t} \rangle \right| = \left| \langle A^i_j, x_t - x'\rangle - \langle A_j^i,  x_t \frac{\lceil k_t / 2 \rceil}{k_t} \rangle \right| =
$$
$$
=  \left| \langle A^i_j ,x_t \frac{\lfloor k_t / 2 \rfloor}{k_t} \rangle - \langle A_j^i,  x'  \rangle \right| \leq \Xi(n_1 \cdot k_t/d, r) \cdot \sqrt{k_t \cdot \Delta_j^i},
$$
which concludes (\ref{eq:nodebound}).

By the triangle inequality, we have:

\begin{align*}
\left| \langle A^i_j, x_{l} \rangle - \left\langle A_j^i, \boldsymbol{1} \frac{1}{k} \right\rangle \right|
&\leq \left| \left\langle A^i_j, x_{l} \right\rangle - \left\langle A_j^i, \frac{x_{l-1}}{k_{l-1}} \right\rangle \right| + \dots + \left| \left\langle A^i_j, \frac{x_{2}}{k_{2}} \right\rangle - \left\langle A_j^i, \frac{x_{1}}{k_{1}} \right\rangle \right| \\
&\overset{\eqref{eq:nodebound}}{\leq} \sqrt{\Delta^i_j} \cdot \left( \Xi \left(\frac{n_1 \cdot k_{l-1}}{d}, r \right) \cdot \frac{\sqrt{k_{l-1}}}{k_l} + \dots + \Xi \left(\frac{n_1 \cdot k_{1}}{d}, r \right) \cdot \frac{\sqrt{k_{1}}}{k_{2}} \right) \\
&\leq \sqrt{\Delta^i_j} \cdot \Xi \left(\frac{n_1}{d}, r\right) \left( \frac{\sqrt{k_{l-1}}}{k_l} + \dots + \frac{\sqrt{k_{1}}}{k_{2}} \right) \\
&\quad + \sqrt{\Delta^i_j} \cdot C \cdot \sqrt{\log_2 (4r)} \left( \frac{\log_2 (k_{l-1}) \cdot \sqrt{k_{l-1}}}{k_l} + \dots + \frac{\log_2(k_1) \cdot \sqrt{k_{1}}}{k_{2}} \right)
\end{align*}

In the last inequality, we used \eqref{eq_xi_def} and the identity $\Xi(a \cdot b, r) = \Xi(a, r) + C \cdot \sqrt{\log_2(4r)} \cdot \log_2 b$. 

We have $k_{t+1} \le  \lceil k_t / 2 \rceil\le \frac 23k_t$ for $k_t\ge 3$. Thus the expresions 
$$\left( \frac{\sqrt{k_{l-1}}}{k_l} + \dots + \frac{\sqrt{k_{1}}}{k_{2}} \right) \ \ \text{and}\ \  \left( \frac{\log_2 (k_{l-1}) \cdot \sqrt{k_{l-1}}}{k_l} + \dots + \frac{\log_2(k_1) \cdot \sqrt{k_{1}}}{k_{2}} \right)$$ can be upper bounded by a constant. Since $n_1 \geq d$, we have $\Xi(\frac{n_1}{d}, r) \geq C \cdot \sqrt{\log_2(4r)}$. Therefore, the final expression can be bounded by

$$
 C' \cdot \Xi \left( \frac{n_1}{d}, r \right) \cdot \sqrt{\Delta^i_j}
$$

for some constant $C' > 0$.
\end{proof}

    Let $A^1 \in [0, 1]^{n_1 \times m}, \ldots, A^k \in [0, 1]^{n_k \times m}$ be agents valuations of objects. Note that $A^i$ is an $(n_i, n_i, \ldots, n_i)$-structured matrix since it has $n_i$ rows. If we directly apply Theorem \ref{thm:delta-struct-division} with $\Delta^i_j = n_i$, then we obtain a division of goods where the agents from the $i$-th group are at most $$C' \cdot \Xi(n_1/n_k, n_1/n_k) \cdot \sqrt{n_i} = \tilde O(\sqrt{n_i})$$ unsatisfied. This is an acceptable error if the sizes of all groups are equal $n_1 = n_2 = \ldots = n_k$. However, in an unbalanced case where, for example, $n_1$ is large and $n_2 = n_3 = \ldots = n_k = 1$, the algorithm in Theorem~\ref{thm:delta-struct-division} is disadvantageous for the large group and, as a result, gives an error of $\Omega(\sqrt{n_1})$, which is too large.

    To deal with this issue, we show that it is possible to rebalance the advantage/disadvantage that the groups get when applying Theorem~\ref{thm:delta-struct-division} by  \textit{gifting}, or simply assigning, a carefully chosen amount of  goods to the larger groups. 

    \begin{theorem}
        \label{thm:gifting}
        Let $C'' > 0$. Let $A^1 \in [0, 1]^{n_1 \times m}, \ldots, A^k \in [0, 1]^{n_k \times m}$ be agents valuations of objects. Then, there exists a partial assignment of goods to groups $v_1, \ldots, v_k \in \{0, 1\}^m$, $\sum v_i\leq \bm{1})$, such that the following properties are satisfied. Let $l = \langle \bm{1}, v_1 + \ldots + v_k \rangle$ denote the number of assigned goods. Let the matrix $\hat{A}^i  \in [0, 1]^{n_i \times (m - l)}$ be obtained from $A^i$ by removing columns that correspond to the assigned items (the ones in $v_1, \ldots, v_k$). Let $\Upsilon := \Xi(\frac{n_1 k}{n}, \frac{n_1 k}{n})$ where $n = n_1 + \ldots + n_k$. Then,
\begin{enumerate}
    \item $l \leq C''  \cdot \sqrt{\Upsilon \cdot n \cdot k}$.
    \item For every group $i$, there is a permutation of agents $p_1^i, \ldots, p_{n_i}^i$ in this group and values $n_1 \geq \Delta^i_1 \geq \ldots \geq \Delta^i_{n_i} \geq n/(\Upsilon \cdot k)$, such that the matrix $B^i$ obtained from rearranging rows of $\hat{A}^i$ according to the permutation $p^i$
    $$
    B^i = \begin{bmatrix}
\hat{A}^i_{p_1^i} \\[1ex]
\hat{A}^i_{p_2^i} \\[1ex]
\vdots \\[1ex]
\hat{A}^i_{p_{n_i}^i}
\end{bmatrix}
    $$
     is $(\Delta^i_{p_1}, \Delta^i_{p_2}, \ldots, \Delta^i_{p_{n_i}})$-structured and for every $j \in [n_i]$,
    $$
    \max\left(\langle B^i_j, v_i \rangle, C'' \cdot \sqrt{\Upsilon \cdot \frac{n}{k}} \right) \geq  C'' \cdot \Upsilon \cdot \sqrt{\Delta^i_j}.
    $$
    \end{enumerate}
    \end{theorem}

Let us comment on the statement of the theorem. The bound on $l$ is chosen so that the gifted items have negative affect that is covered by the error term. The condition of being  $(\Delta^i_j)$-structured is a technical condition needed for an application of the discrepancy results and essentially allows to bound the row norms of a matrix and its submatrices. The maximum in the statement encodes the two scenarios. One option is that a given row has small $\ell_1$-norm. In this case, we do not need to rebalance it for the application of Theorem~\ref{thm:delta-struct-division} and may ignore the contribution of gifting. The other option is that a row has large $\ell_1$-norm, and Theorem~\ref{thm:delta-struct-division} gives a large error on this row. Then the gifting completely covers this possible error. 


We first derive the main theorem from these ingredients. To this end, we need the following auxiliary technical lemma.
\begin{lemma}
\label{lem:trivbound}
There exists a constant $C'' > 0$, such that for all $r \geq 1$,

$$
\Xi(\Xi(r, r) \cdot r, \Xi(r, r) \cdot r) \leq C'' \cdot \Xi(r, r).
$$
\end{lemma}

\begin{proof}

\

\begin{align*}
\Xi(\Xi(r, r) \cdot r, \Xi(r, r) \cdot r) &= C \cdot (\log_2(r) + \log_2(\Xi(r, r)) + 2) \sqrt{\log_2 (r) + \log_2(\Xi(r, r)) + 2} \\
&= O\left((\log_2(r) + 2) \sqrt{\log_2 (r) + 2}\right), \quad \text{since } 1 \leq \Xi(r, r) = O(r) \\
&\leq C'' \cdot \Xi(r, r)
\end{align*}

for some $C'' > 0$.
\end{proof}

Now we are ready to combine the procedure in Theorem \ref{thm:delta-struct-division} and the gifting of Theorem~\ref{thm:gifting} to establish the upper bound in Theorem \ref{thm:main}.

\begin{proof}[Proof of the upper bound in Theorem \ref{thm:main}]

    Let $C'$ be the constant from Theorem \ref{thm:delta-struct-division}. Denote $n = n_1 + \ldots + n_k$ and let $\Upsilon := \Xi(\frac{n_1 k}{n}, \frac{n_1 k}{n})$. By Lemma \ref{lem:trivbound}, there is a constant $C'' > 0$, such that $C' \cdot \Xi(\Upsilon \cdot \frac{n_1 \cdot k}{n}, \Upsilon \cdot \frac{n_1 \cdot k}{n}) \leq C'' \cdot \Xi(\frac{n_1 \cdot k}{n}, \frac{n_1 \cdot k}{n})$.

    Apply Theorem \ref{thm:gifting} with constant $C''$ to obtain a partial assignment $v_1, \ldots, v_k \in \{0, 1\}^m$ and matrices $B^1, \ldots, B^k$ that are $(\Delta^1_{p^1_j}), (\Delta^2_{p^2_j}), \ldots, (\Delta^k_{p^k_j})$ structured, respectively. Let $l = \langle \bm{1}, v_1 + \ldots + v_k \rangle$ denote the number of goods that were assigned by the application of Theorem~\ref{thm:gifting}.

    Now, we apply Theorem \ref{thm:delta-struct-division} to matrices $B^1, \ldots, B^k$ with $r = \Upsilon \cdot  \frac{ n_1  \cdot k}{n}$ and $d = n/(\Upsilon \cdot k)$ to divide the remaining $m - l$ objects among $k$ groups. Note that $dr = n_1$, and the matrices $B^i$ obviously satisfy $\Delta^i_j\le n_1$. Thus, such a choice of $r,d$ is valid. Denote by $u_1, \ldots, u_k \in \{0, 1\}^{m - l}$ the division given by the theorem. We are guaranteed that for the $j$-th agent from the $i$-th group, it holds
    $$
    \langle B^i_{p^i_j}, u_i \rangle  \geq  \langle B^i_{p^i_j},\bm{1} / k \rangle - C' \cdot \Xi \left(\Upsilon \frac{n_1 k}{n}, \Upsilon \frac{n_1 k}{n} \right) \cdot \sqrt{\Delta^i_{p^i_j}} \geq \langle B^i_{p^i_j},\bm{1} / k \rangle - C'' \cdot \Upsilon \cdot \sqrt{\Delta^i_{p^i_j}}.
    $$
    Observe that $\max\left(\langle A^i_j, v_i \rangle, C'' \cdot \sqrt{\Upsilon \cdot \frac{n}{k}} \right) \geq  C'' \cdot \Upsilon \cdot \sqrt{\Delta_j^i}$ implies $\langle A^i_j, v_i \rangle \geq  C'' \cdot \Upsilon \cdot \sqrt{\Delta_j^i} - C'' \cdot \sqrt{\Upsilon \cdot \frac{n}{k}}$. In addition, since at most $l$ objects are gifted, we have $\langle B_{p_j^i}^i, \bm{1}/k \rangle \geq \langle A_j^i, \bm{1}/k\rangle - l/k \geq \langle A_j^i, \bm{1}/k\rangle - C'' \cdot \sqrt{\Upsilon \cdot \frac{n}{k}}$.

    Therefore, combined with the gifted objects $v_i$, the $j$-th agent from the $i$-th group values the assigned objects at least by

    $$
        \langle A^i_j, v_i \rangle  + \langle B^i_{p_j^i}, u_i \rangle  \geq \langle A^i_j, \bm{1}/k \rangle - 2 \cdot C'' \cdot \sqrt{\Upsilon \cdot \frac{n}{k}}.
    $$

    Thus, we have obtained the desired $2 \cdot C'' \cdot \sqrt{\Upsilon \cdot \frac{n}{k}}$ bound on PROP'.

\end{proof}

In order to complete the proof of the main theorem, we are left to  prove Theorem~\ref{thm:gifting}. This is done in the following section.

\subsection{Proof of Theorem \ref{thm:gifting}}\label{sec211}

The overall assignment is simply a combination of partial assigments obtained for each group individually using the following theorem.

\begin{theorem}
     \label{thm:gifting_individual}
     Let $C'' > 0$ be a real, and let $n > 0, n_1 > 0, k > 0$ be integers, such that $n \geq n_1 \geq n/k$. Let $A \in [0, 1]^{n_s \times m}$ be agents valuations of objects for the $s$-th group with $n_s \leq n_1$ agents. Then, there exists a partial assignment of goods $v \in \{0, 1\}^m$ , such that the following properties are satisfied. Let $l = \langle \bm{1}, v \rangle$ denote the number of assigned goods, and let $\Upsilon := \Xi(\frac{n_1 \cdot k}{n}, \frac{n_1 \cdot k}{n})$. Let matrix $A' \in [0, 1]^{n_s \times (m - l)}$ be obtained from $A$ by removing columns that correspond to the assigned items (ones in $v$). Then,
\begin{enumerate}
    \item $l \leq C'' \cdot n_s \cdot \sqrt{\Upsilon \cdot \frac{k}{n}}$
    \item There is a permutation of agents $p_1, \ldots, p_{n_s}$ and values $n_1 \geq \Delta_1 \geq \ldots \geq \Delta_{n_s} \geq n/(\Upsilon \cdot k)$, such that the matrix obtained from rearranging rows of $A'$ according to the permutation $p$
    $$
    B = \begin{bmatrix}
A'_{p_1} \\[1ex]
A'_{p_2} \\[1ex]
\vdots \\[1ex]
A'_{p_{n_s}}
\end{bmatrix}
    $$
     is $(\Delta_{p_1}, \Delta_{p_2}, \ldots, \Delta_{p_{n_s}})$-structured and for every $j \in [n_s]$,
    $$
    \max\left(\langle B_j, v \rangle, C'' \cdot \sqrt{\Upsilon \cdot \frac{n}{k}} \right) \geq  C'' \cdot \Upsilon \cdot \sqrt{\Delta_j}.
    $$
    \end{enumerate}
\end{theorem}

Let us show how to combine these partial assignments and prove Theorem \ref{thm:gifting}.

\begin{proof}[Proof of Theorem \ref{thm:gifting} via   Theorem \ref{thm:gifting_individual}]
    The procedure has $k$ iterations, where at the $i$-th iteration we select goods for the $i$-th group. \begin{itemize}\item As an input, we get the matrices $\hat A^1(i),\ldots, \hat A^k(i), B^1(i),\ldots, B^{i-1}(i)$. Initially (iteration $1$), we have $\hat A^j(1) = A^j$.
    \item We apply Theorem~\ref{thm:gifting_individual} to the $i$-th group with the assignment matrix $\hat A^i(i)$ playing the role of $A$ and  with $s = i$ to obtain a partial assignment $v_i := v$, a permutation $p^i := p$, and a matrix $B^i(i):= B$.
    \item We then obtain $\hat A^1(i+1),\ldots, \hat A^k(i+1), B^1(i+1),\ldots, B^{i}(i+1)$ by removing the columns that correspond to the items that were gifted at this step from the matrices $\hat A^1(i),\ldots, \hat A^k(i)$, $B^1(i),\ldots, B^{i}(i)$.  
\item After the last step $k$, we output $\hat A^j:=\hat A^j(k+1)$  and $B^j=B^j(k+1)$.
\end{itemize}

    Note that in total the procedure assigns at most $C'' \cdot (n_1 + \ldots + n_k) \cdot \sqrt{\Upsilon \cdot k/n} = 4 \cdot \sqrt{\Upsilon \cdot n \cdot k}$ objects. Moreover, for every $i$, the matrix $B^i$ is $(\Delta^i_{p_j^i})$ structured, since it was $(\Delta^i_{p_j^i})$-structured at step $i$ and this property is hereditary w.r.t. removing columns. 
    At last, for every $i \in [k], j \in [n_i]$, we have $\langle B^i_j,v_i\rangle=\langle B^i_j(i),v_i\rangle$, and 
    $$
    \max\left(\langle B^i_j, v_i \rangle, C'' \cdot \sqrt{\Upsilon \cdot \frac{n}{k}} \right) \geq  C'' \cdot \Upsilon \cdot \sqrt{\Delta_j^i},
    $$
    since it was guaranteed by Theorem \ref{thm:gifting_individual}.
\end{proof}

The rest of the section is devoted to the proof of Theorem \ref{thm:gifting_individual}. We begin with the description of the following algorithm that iteratively constructs the partial assignment $v$. See Algorithm \hyperref[alg:gift]{1} for the pseudocode.
\

\

\hrule

\

\paragraph{Gifting algorithm.} The algorithm starts with an empty assignment $v := 0^m$ and sequentially selects an object that is given to the group. During its execution, the algorithm also maintains a set $S$ of \textit{active} agents. Initially, every agent is active, i.e. $S = \{1, 2, \ldots, n_s\}$.

On each iteration, the algorithm finds a column $j \in [m]$, such that $v_j = 0$ and the sum of active agents valuations of the $j$-th object is maximized: $\sum_{i \in S} A_{i,j} \to \max$. Let $\lambda := \sum_{i \in S} A_{i,j}$ denote the column sum over active agents. We iterate over active agents $i \in S$ and look if $\langle A_i, v \rangle \geq C'' \cdot \Upsilon \cdot \sqrt{\lambda}$. If this holds, then we consider the $i$-th agent to already be fulfilled and remove him from the set of active agents $S \leftarrow S \setminus \{i \}$. If at least one agent is removed after the filtration, we recalculate $j$ and $\lambda$ that maximizes the column sum over the modified set of active agents and perform the filtration again.

After that, at the end of the iteration, if $\lambda < \frac{1}{\Upsilon} \cdot \frac{n}{k}$ or $S = \varnothing$, then the algorithm terminates. Otherwise, we give the $j$-th item to the group: $v_j \leftarrow 1$ and continue with the next iteration.

\

\hrule
\

\

\begin{algorithm}
\caption{Gifting Algorithm}
\begin{algorithmic}[1]
\label{alg:gift}
\State $v \gets 0^m$ \Comment{Initially no object is gifted}
\State $S \gets \{1, 2, \dots, n_s\}$ \Comment{Initially every agent is active}
\For{$t = 1, 2, \dots$}
    \State $j \gets \arg\max_{j : v_j = 0} \sum_{i \in S} A_{i,j}$ \Comment{Find item most desirable by the active agents}
    \State $\lambda \gets \sum_{i \in S} A_{i, j}$

    \State $\text{atLeastOneRemoved} \gets \text{False}$
    \For{$i \in S$}
        \If{$\langle A_i, v \rangle \geq C'' \cdot \Upsilon \cdot \sqrt{\lambda}$}
            \State $S \gets S \setminus \{i\}$ \Comment{Make the $i$-th agent inactive}
            \State $\text{atLeastOneRemoved} \gets \text{True}$
        \EndIf
    \EndFor

    \If{atLeastOneRemoved}
        \State \textbf{goto} \textbf{Step 4} \Comment{Update $j$ and $\lambda$ and repeat filtration}
    \EndIf

    \If{$\lambda < \frac{1}{\Upsilon} \cdot \frac{n}{k}$ \textbf{or} $S = \varnothing$}
        \State \textbf{terminate and return} $v$  \Comment{$\lambda$ is small or all agents are fullfilled}
    \Else
        \State $v_j \gets 1$ \Comment{Give the $j$-th object to the group}
    \EndIf
\EndFor
\end{algorithmic}
\end{algorithm}

Denote $T = \|v\|_1$ to be the number of items that the algorithm has gifted. Note that the algorithm made $T + 1$ iterations in total. The $T$-th iteration was the last when an item was added to the group. For $t=1 \ldots T + 1$, denote $S_t$ to be the set of active agents at the end of the $t$-th iteration, and denote by $\lambda_t$ the column sum over active agents at the end of the $t$-th iteration.

Let us note that $\lambda_T \geq \frac{1}{\Upsilon} \cdot \frac{n}{k}$ since the algorithm did not terminate at $T$-th iteration. Moreover, observe that $\lambda_t$ is monotone in $t$, $\lambda_1 \geq \lambda_2 \geq \lambda_3 \geq \ldots \geq \lambda_T$, since the space of active agents and columns over which we maximize and take the sum only shrinks.

Let $t_i$ be the iteration when the $i$-th agent is removed from the active set. If the agent $i$ is never removed from the active set, we set $t_i := T + 1$.

\

Let us show that $\|v\|_1$ is bounded.

\begin{claim}
    The algorithm described above assigns no more than $C'' \cdot n_s \cdot \sqrt{\Upsilon \cdot k/n}$ objects to the group. That is, $\|v\|_1 \leq C'' \cdot n_s \cdot \sqrt{\Upsilon \cdot k/n}$.
\end{claim}

\begin{proof}
   For $t=1\ldots T$, let us denote by $j_t$  the object given to the group on the $t$-th iteration.
   Let us calculate the following sum in two different ways:
$$
\sum_{t=1}^T \sum_{i \in S_t} \frac{A_{i, j_t}}{\lambda_t}.
$$
Firstly, note that for every $t=1 \ldots T$, we have $\sum_{i \in S_t} A_{i, j_t} = \lambda_t$ by the definition of $\lambda_t$. Therefore, on the one hand
$$
\sum_{t=1}^T \sum_{i \in S_t} \frac{A_{i, j_t}}{\lambda_t} = \sum_{t=1}^T \frac{\lambda_t}{\lambda_t} = T.
$$

On the other hand, we can swap the summations to sum over rows first:

\begin{align}
    \sum_{t=1}^T \sum_{i \in S_t} \frac{A_{i, j_t}}{\lambda_t} &= \sum_{i=1}^{n_s} \sum_{t=1}^{t_i - 1} \frac{A_{i, j_t}}{\lambda_t} && \color{gray} \begin{tabular}[t]{@{}l@{}} reordering summations, $A_{i, j_t}$ contributes \\ to $\lambda_t$ for $t \leq t_i - 1$ \end{tabular} \nonumber \\[10pt]
    &\leq \sum_{i=1}^{n_s} \frac{1}{\lambda_{t_i - 1}} \sum_{t=1}^{t_i - 1} A_{i, j_t} && \color{gray} \text{monotonicity } \lambda_t \geq \lambda_{t_i - 1} \nonumber \\[10pt]
    &\leq \sum_{i=1}^{n_s} \frac{1}{\lambda_{t_i - 1}} \cdot \left( C'' \cdot \Upsilon \cdot \sqrt{\lambda_{t_i - 1}} \right) && \color{gray} \begin{tabular}[t]{@{}l@{}} $\sum_{t=1}^{t_i - 1} A_{i, j_t} < C'' \cdot \Upsilon \cdot \sqrt{\lambda_{t_{i} - 1}}$ holds. \\ Otherwise, the $i$-th agent would be removed \\ from the active set on the $(t_i - 1)$-th iteration \end{tabular} \nonumber \\[10pt]
    &= \sum_{i=1}^{n_s} \frac{1}{\sqrt{\lambda_{t_i - 1}}} C'' \cdot \Upsilon && \color{gray} \text{simplification} \nonumber \\[10pt]
    &\leq \sum_{i=1}^{n_s} \frac{1}{\sqrt{\lambda_{T}}} C'' \cdot \Upsilon && \color{gray} \text{monotonicity } \lambda_{t_i - 1} \geq \lambda_T \nonumber \\[10pt]
    &\leq C'' \cdot n_s \cdot \sqrt{\Upsilon \cdot \frac{k}{n}}  && \color{gray} \text{using } \lambda_T \geq \frac{1}{\Upsilon} \cdot \frac{n}{k}
\end{align}

Therefore, we obtain the desired bound

$$
\|v\|_1 = T \leq C'' \cdot n_s \cdot \sqrt{\Upsilon \cdot \frac{k}{n}}.
$$

\end{proof}
To conclude the proof of Theorem \ref{thm:gifting_individual}, let us define the matrix $B$, the permutation of agents $p$, and the sequence $\Delta_i$.

We define $\Delta_i$ to be the value of $\lambda$ right at the moment when the $i$-th agent is removed from the set $S$. If the agent is never removed from the set $S$, we set $\Delta_i := \max\left(\lambda_{T+1}, \frac{1}{\Upsilon} \cdot \frac{n}{k} \right)$. Let $p_1, p_2, \ldots, p_{n_s}$ be the sequence at which agents are removed from the active set ($p_1$ is the first agent that is removed, $p_2$ is the second, and so on). The agents that are not removed from the active set by the end of the algorithm are placed at the end of the sequence $p$ in any order. Observe that $n_1 \geq \Delta_{p_1} \geq \Delta_{p_2} \geq \ldots \geq \Delta_{p_{n_s}} \geq n/(k \cdot \Upsilon)$ since $\lambda$ can only decrease over time.

Let $A'$ be matrix obtained from $A$ by removing all columns corresponding to the gifted items (ones in $v$). Define the matrix $B$ by arranging the rows of matrix $A'$ in the following sequence

$$
B = \begin{bmatrix}
A'_{p_1} \\[1ex]
A'_{p_2} \\[1ex]
\vdots \\[1ex]
A'_{p_{n_s}}
\end{bmatrix}
$$

Let us show that the matrix $B$ is top-down structured.

\begin{claim}
\label{claim:gift-struct}
    The matrix $B$ is $(\Delta_{p_1}, \Delta_{p_2}, \ldots, \Delta_{p_{n_s}})$-structured.
\end{claim}
\begin{proof}
Consider the $i$-th row of the matrix $B$. We consider two cases whether the $p_i$-th agent was removed from the set $S$ or not.

If the $p_i$-th agent is never removed from the set $S$, then $p_{i}, p_{i+1}, \ldots, p_{n_s} \in S_{T+1}$. By definition, $\lambda_{T+1}$ is the maximum column sum over rows corresponding to agents from $S_{T+1}$. Therefore, $\Delta_{p_i} = \max\left(\lambda_{T+1}, \frac{n}{\Upsilon \cdot k} \right)$ is also an upper bound on the column sum of the submatrix $B_{i,\ldots, n_s}$.

Now, suppose the $p_i$-th agent is removed from the set $S$ at some point. Consider the moment just before $p_i$-th agent is removed from $S$. The agents $p_{i}, p_{i+1}, \ldots, p_{n_s}$ all lie in the set $S$. We also have that the current value of $\lambda$ is an upper bound on the maximum column sum over rows corresponding to agents from $S$. Therefore, $\Delta_{p_i} = \lambda$ is also an upper bound on the column sum of the submatrix $B_{i, \ldots, n_s}$.

Thus, the matrix $B$ is $(\Delta_{p_1}, \ldots, \Delta_{p_{n_s}})$-structured.
\end{proof}

Finally, let us show that every agent values gifted objects at the desired threshold.

\begin{claim}
\label{claim:necbound}
For every agent $i \in [n_s]$, it holds that
    \begin{equation}
    \max\left(\langle A_i, v \rangle, C'' \cdot \sqrt{\Upsilon \cdot \frac{n}{k}} \right) \geq  C'' \cdot \Upsilon \cdot \sqrt{\Delta_i}.   \label{eq:desbound}
    \end{equation}

\end{claim}
\begin{proof}
Let us condition on whether the $i$-th agent was removed from the set $S$ or not.

If it was not removed, then $\Delta_i = \lambda_{T+1} < \frac{1}{\Upsilon} \cdot \frac{n}{k}$. Therefore, $\sqrt{\Upsilon \cdot \frac{n}{k}} \geq \Upsilon \cdot \sqrt{\Delta_i}$.

If it was removed, then it could happen only if $\langle A_i, v \rangle \geq C'' \cdot \Upsilon \cdot \sqrt{\lambda}$ was true at that moment. Since $\Delta_i = \lambda$, we get the desired bound.

Therefore, in both cases we can conclude (\ref{eq:desbound}).
\end{proof}

The claims \ref{claim:gift-struct} and \ref{claim:necbound} together conclude the proof of Theorem \ref{thm:gifting_individual}.

\section{Lower Bound in Theorem~\ref{thm:main2}} \label{sec_lower_bound}

In this section, we prove that $\disc^k_>=\Omega(\sqrt{n/k})$. We start with a trivial lower bound on $\oasymdiscm$. 

\begin{lemma}
\label{low:trivial}
    $$
    \disc^k_>(n_1, \ldots, n_k) \geq \frac{k-1}{k}.
    $$
\end{lemma}
\begin{proof}
    Let $k \geq 2$ be the number of groups. Let $n_1, \ldots, n_k > 0$ be group sizes, and let $n = n_1 + \ldots + n_k$ be the total number of agents.

Suppose there are $k-1$ objects, and each of them is valued by $1$ by every agent. Then, in any allocation at least one group receives no objects. The fair share for this group is $\frac{k - 1}{k}$. Therefore, agents from this group are at least $\frac{k-1}{k}$ unsatisfied.
\end{proof}

To prove our lower bound, we rely on the following lower bound on weighted discrepancy.

\begin{theorem}[Theorem 3.2 from \citep{DBLP:conf/sosa/ManurangsiM26}]
    \label{thm:hardmatrix}
    For any integers $r > 0$ and $k > 0$, there exists a matrix $A \in [0, 1]^{r \times (r \cdot k)}$, such that for any $v \in \{0, 1\}^m$, it holds

    $$
    \| A (v - \bm{1}/k)\|_{\infty} \geq \sqrt{r - 1}/16.
    $$
\end{theorem}

The theorem guarantees that the maximum of the absolute value of deviation of $\langle A_i, (v - \bm{1}/k)\rangle$ from $0$ is at least $\sqrt{r - 1}/16$. If the deviation is downwards then we immediately infer that  some agent is unsatisfied. However, an upward deviation is only welcomed by an agent and thus does not immediately imply any lower bounds for group division. To deal with this, we introduce agents with complementary valuations. This converts an upward deviation into a downward one. Similarly to \citep{DBLP:conf/sosa/ManurangsiM26}, we derive the following result.

\begin{lemma}
   \label{lem:hardcompl}
    For any $r \geq 1$, there exists $m_0$, such that for any $m \geq m_0$ there exists a matrix $B \in \{0, 1\}^{2r \times m}$ that satisfies the following. For any $v \in \{0, 1\}^m$, there exists $i$, such that
\begin{equation}
\label{eq:fairbound}
\langle B_i, v \rangle \leq \langle B_i, \bm{1}/k\rangle - \frac{\sqrt{r - 1}}{16} + \max(0, \|v\|_1 - m/k)
\end{equation}

\end{lemma}
\begin{proof}
    Let $A \in \{0, 1\}^{r \times m_0}$ be a matrix given by Theorem \ref{thm:hardmatrix}. For $m \geq m_0$, let a matrix $A' \in \{0, 1\}^{r \times m}$ be obtained from $A$ by padding extra columns with zeroes. Note that the matrix $A'$ still satisfies the property from Theorem \ref{thm:hardmatrix} since zero columns do not change the product.

Define $C$ to be the matrix of rows complementary to $A'$:  $C_{i, j} = 1 - A'_{i, j}$ for all $i, j$. Let $B$ be a vertical concatenation of matrices $A'$ and $C$:

 $$
B = \begin{bmatrix}
A' \\[1ex]
C \\[1ex]
\end{bmatrix}.
$$
Consider a vector $v \in \{0, 1\}^m$ and let us show that the matrix $B$ satisfies the required property.

By Theorem \ref{thm:hardmatrix}, for any $v \in \{0, 1\}^m$, there exists $i$, such that

$$
|\langle A'_i, v - \bm{1}/k \rangle| \geq \frac{\sqrt{r - 1}}{16}.
$$
Let us consider the two cases how the absolute value can be expanded.

If $-\langle A'_i, v - \bm{1}/k \rangle \geq \sqrt{r - 1}/16$ then $\langle A'_i, v  \rangle \leq \langle A'_i, \bm{1}/k \rangle - \sqrt{r - 1}/16$, which implies (\ref{eq:fairbound}).

If  $\langle A'_i, v - \bm{1}/k \rangle \geq \sqrt{r - 1}/16$, then
$$\langle 1-A'_i, v - \bm{1}/k \rangle \le -\sqrt{r - 1}/16 + \langle \bm{1}, v v- \bm{1}/k \rangle =-\sqrt{r - 1}/16 +  \|v\|_1 - m/k.$$ Since $1-A'_i$ is a row of $B$, we get the desired conclusion.
\end{proof}

\subsection{Proof of the Lower Bound}

Lemma \ref{low:trivial} gives  a lower bound $\disc^k_>(n_1, \ldots, n_k) \geq \frac12$. This allows us to assume that $n \geq 2^{14} \cdot k$. Denote $r = \lfloor n/(4 \cdot k) \rfloor \geq 2^{12}$, and let $m_0$ be the integer from Lemma \ref{lem:hardcompl} when applied with $r$. Let $M = n \cdot m_0$ be the number of objects. We design agents valuations for the $i$-th group as follows.

\begin{itemize}
    \item If $n_i < 2r$, then we set all of the agents from the $i$-th group to valuate any object by $1$.
    \item If $n_i \geq 2r$, then we divide all agents into subgroups of size $2 r$ with $<2r$ agents left in no subgroup. In total, this gives $s_i = \lfloor n_i/(2r_i) \rfloor$ subgroups. Let $m_i = \lfloor M/s_i \rfloor \geq m_0$, and let the matrix $B \in \{0, 1\}^{2r \times m_i}$ be given by Lemma \ref{lem:hardcompl} . For the $j$-th subgroup we set the valuations of the objects $(j-1) \cdot m_i + 1,  \ldots, j \cdot m_i$ according to the matrix $B$. All other objects are valuated by $0$ by the agents of this subgroup. See Figure \ref{fig:lowerbound_construction} for an illustration of the chosen valuations.
\end{itemize}
Let us give some intuition behind the construction. Note that $r$ is essentially an average size of a group. Thus, large groups get valuations in the spirit of the construction from \cite{DBLP:conf/sosa/ManurangsiM26}, while small groups evaluate everything to $1$. The evaluation of small groups thus forces that roughly a correct proportion of goods go to the large groups. There, the goods are  partitioned between subgroups. This way, we are able to apply the bound of Lemma~\ref{lem:hardcompl} inside a subgroup, rather than inside the whole group.\\ 

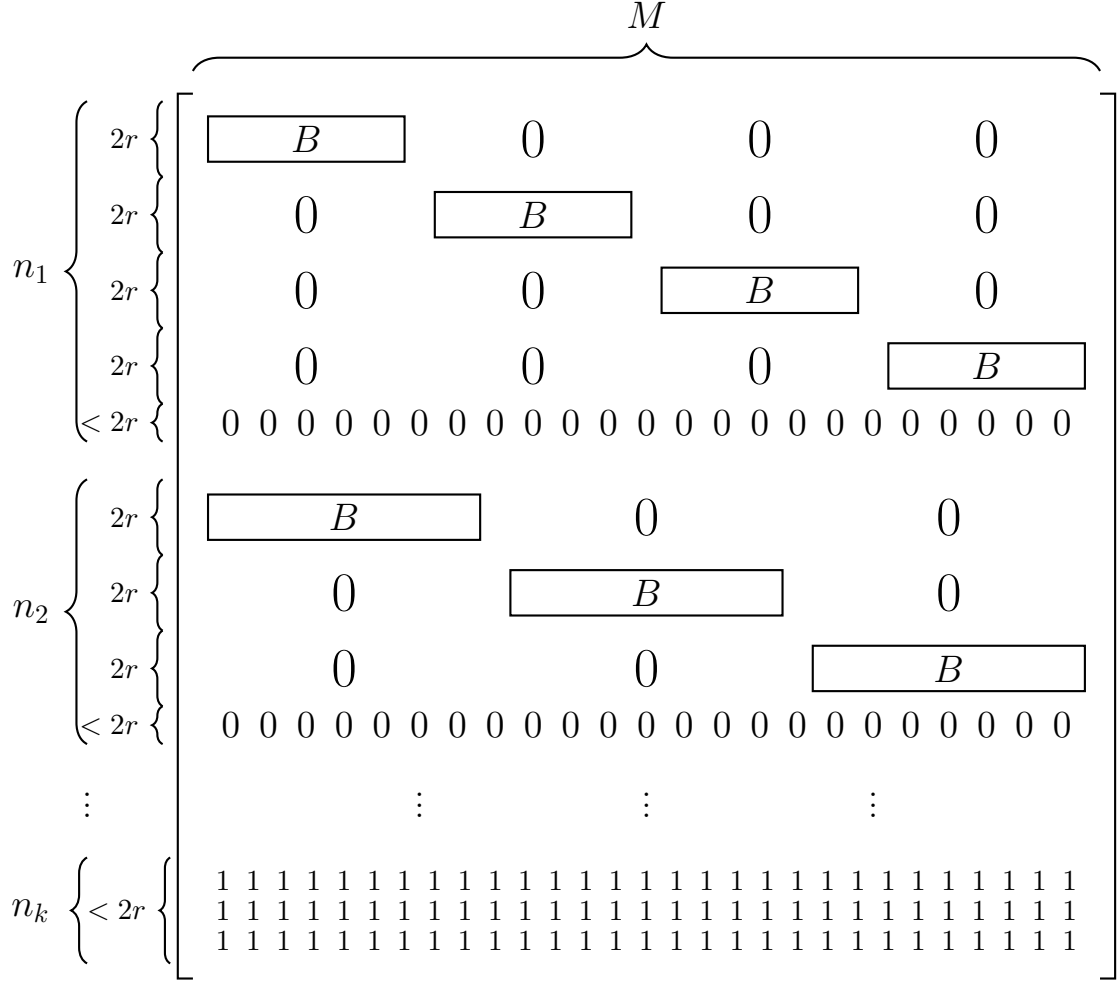
\begin{figure}[h]
    \centering
\begin{tikzpicture}[x=1cm, y=1cm, yscale=-1]

    \draw [thick, decorate, decoration={brace, amplitude=10pt}]
        (0, -0.4) -- (12, -0.4) node [midway, above=12pt] {\Large $M$};

    \draw [thick] (0, -0.1) -- (-0.2, -0.1) -- (-0.2, 11.6) -- (0, 11.6);
    \draw [thick] (12, -0.1) -- (12.2, -0.1) -- (12.2, 11.6) -- (12, 11.6);

    \draw[thick] (0.2, 0.2) rectangle (2.8, 0.8) node[midway] {\Large $B$};
    \draw[thick] (3.2, 1.2) rectangle (5.8, 1.8) node[midway] {\Large $B$};
    \draw[thick] (6.2, 2.2) rectangle (8.8, 2.8) node[midway] {\Large $B$};
    \draw[thick] (9.2, 3.2) rectangle (11.8, 3.8) node[midway] {\Large $B$};

    \foreach \x/\y in {
        4.5/0.5, 7.5/0.5, 10.5/0.5,
        1.5/1.5, 7.5/1.5, 10.5/1.5,
        1.5/2.5, 4.5/2.5, 10.5/2.5,
        1.5/3.5, 4.5/3.5, 7.5/3.5} {
        \node at (\x, \y) {\huge $0$};
    }

    \foreach \x in {0.5, 1.0, ..., 11.5} {
        \node at (\x, 4.25) {\Large $0$};
    }

    \draw [thick, decorate, decoration={brace, amplitude=4pt}] (-0.4, 1) -- (-0.4, 0) node [midway, left=5pt] {$2r$};
    \draw [thick, decorate, decoration={brace, amplitude=4pt}] (-0.4, 2) -- (-0.4, 1) node [midway, left=5pt] {$2r$};
    \draw [thick, decorate, decoration={brace, amplitude=4pt}] (-0.4, 3) -- (-0.4, 2) node [midway, left=5pt] {$2r$};
    \draw [thick, decorate, decoration={brace, amplitude=4pt}] (-0.4, 4) -- (-0.4, 3) node [midway, left=5pt] {$2r$};
    \draw [thick, decorate, decoration={brace, amplitude=4pt}] (-0.4, 4.5) -- (-0.4, 4) node [midway, left=5pt] {$<2r$};

    \draw [thick, decorate, decoration={brace, amplitude=8pt}] (-1.4, 4.5) -- (-1.4, 0) node [midway, left=10pt] {\Large $n_1$};

    \begin{scope}[yshift=5cm]
        \draw[thick] (0.2, 0.2) rectangle (3.8, 0.8) node[midway] {\Large $B$};
        \draw[thick] (4.2, 1.2) rectangle (7.8, 1.8) node[midway] {\Large $B$};
        \draw[thick] (8.2, 2.2) rectangle (11.8, 2.8) node[midway] {\Large $B$};

        \foreach \x/\y in {
            6.0/0.5, 10.0/0.5,
            2.0/1.5, 10.0/1.5,
            2.0/2.5, 6.0/2.5} {
            \node at (\x, \y) {\huge $0$};
        }

        \foreach \x in {0.5, 1.0, ..., 11.5} {
            \node at (\x, 3.25) {\Large $0$};
        }

        \draw [thick, decorate, decoration={brace, amplitude=4pt}] (-0.4, 1) -- (-0.4, 0) node [midway, left=5pt] {$2r$};
        \draw [thick, decorate, decoration={brace, amplitude=4pt}] (-0.4, 2) -- (-0.4, 1) node [midway, left=5pt] {$2r$};
        \draw [thick, decorate, decoration={brace, amplitude=4pt}] (-0.4, 3) -- (-0.4, 2) node [midway, left=5pt] {$2r$};
        \draw [thick, decorate, decoration={brace, amplitude=4pt}] (-0.4, 3.5) -- (-0.4, 3) node [midway, left=5pt] {$<2r$};

        \draw [thick, decorate, decoration={brace, amplitude=8pt}] (-1.4, 3.5) -- (-1.4, 0) node [midway, left=10pt] {\Large $n_2$};
    \end{scope}

    \node at (3, 9.2) {\Large $\vdots$};
    \node at (6, 9.2) {\Large $\vdots$};
    \node at (9, 9.2) {\Large $\vdots$};
    \node at (-1.4, 9.2) {\Large $\vdots$};

    \begin{scope}[yshift=10cm]
        \foreach \y in {0.3, 0.7, 1.1} {
            \foreach \x in {0.4, 0.8, 1.2, 1.6, 2.0, 2.4, 2.8, 3.2, 3.6, 4.0, 4.4, 4.8, 5.2, 5.6, 6.0, 6.4, 6.8, 7.2, 7.6, 8.0, 8.4, 8.8, 9.2, 9.6, 10.0, 10.4, 10.8, 11.2, 11.6} {
                \node at (\x, \y) {$1$};
            }
        }

        \draw [thick, decorate, decoration={brace, amplitude=6pt}] (-1.4, 1.4) -- (-1.4, 0) node [midway, left=10pt] {\Large $n_k$};
        \draw [thick, decorate, decoration={brace, amplitude=4pt}] (-0.3, 1.4) -- (-0.3, 0) node [midway, left=5pt] {$<2r$};
    \end{scope}

\end{tikzpicture}
\caption{The valuations of objects for each group. The big zeros correspond to a matrix of $0$s. The valuations of groups with fewer than $2r$ agents are set to $1$. If a group $i$ has at least $2r$ agents, then it is partitioned into subgroups of size $2r$ where each subgroup values some subset of objects according to matrix $B$. The remaining $<2r$ agents get arbitrary valuations (for concreteness, we set them to $0$).}
\label{fig:lowerbound_construction}
\end{figure}

We go on to the proof. Arguing indirectly, assume that there is an allocation $v_1, \ldots, v_k$, such that every agent is at most $2^{-15} \cdot \sqrt{n/k}$ unsatisfied.

Define $d_i = \sum_{j=1}^{M} v_{i, j} - \frac{M}{k}$ to be the difference between the number of allocated goods and the fair share for the $i$-th group.  Note that $\sum_{i=1}^k d_i = 0$. Let us estimate $\sum_{n_i \geq 2r} d_i$ and $\sum_{n_i < 2r} d_i$ separately.

Consider a group $i$ for which $n_i \geq 2r$. We have $s_i = \lfloor n_i / 2r \rfloor$ subgroups in this group. The $j$-th subgroup has non-zero valuations for the segment of objects $(j - 1) \cdot m_i + 1, \ldots, j \cdot m_i$. By the pigeonhole principle, there is a subgroup $j$, such that $\sum_{l=(j-1) \cdot m_i + 1}^{j \cdot m_i} v_{i, l} - \frac{M}{k \cdot s_i} \leq \frac{d_i}{s_i}$. By (\ref{eq:fairbound}), there is an agent in this subgroup that is at least $\frac{\sqrt{r - 1}}{16} - \max(0, d_i / s_i + 1/k)$ unsatisfied (here we used $|m_i - M/s_i| \leq 1$).  If $d_i/s_i < \frac{\sqrt{r - 1}}{32}$, then we get a contradiction since this agent's unsatisfaction is at least

$$
\frac{\sqrt{r - 1}}{16} - \frac{\sqrt{r - 1}}{32} - \frac{1}{2} \geq \frac{\sqrt{r - 1} - 16}{32} > 2^{-15} \cdot \sqrt{\frac{n}{k}}.
$$
Thus, we have $d_i \geq s_i \cdot \sqrt{r - 1}/32 \geq 2^{-11} \cdot s_i  \cdot \sqrt{n/k}$. Therefore,

\begin{align*}
\sum_{n_i \geq 2r} d_i &\geq \sum_{n_i \geq 2r} 2^{-13} \cdot \frac{n_i}{r} \cdot \sqrt{\frac{n}{k}} \\
&\geq \sum_{n_i \geq 2r} 2^{-13} \cdot n_i \cdot \sqrt{\frac{k}{n}} \\
&\geq 2^{-13} \cdot \sqrt{\frac{k}{n}} \left(\sum_{i=1}^k n_i - \sum_{n_i < 2r} n_i \right) \\
&\geq 2^{-13} \cdot \sqrt{\frac{k}{n}} \left(\sum_{i=1}^k n_i - 2r \cdot k \right) \\
&\geq 2^{-13} \cdot \sqrt{\frac{k}{n}} \left(n - \frac{n}{2}\right) \\
&\geq 2^{-14} \cdot \sqrt{n \cdot k}.
\end{align*}

Consider any agent from group $i$ for which $n_i < 2r$. Recall that it evaluates every good to $1$, and thus $d_i$ is exactly the measure of how unsatisfied he is. Thus, $d_i$ cannot be less than $-2^{-15} \cdot \sqrt{n/k}$. Therefore, we have

$$
\sum_{n_i < 2r} d_i \geq  -2^{-15} \cdot \sqrt{n \cdot k}.
$$
Combining two bounds, we have
$$
\sum_{i=1}^k d_i \geq (2^{-14} - 2^{-15}) \cdot \sqrt{n \cdot k} > 0,
$$
reaching a contradiction.

\bibliographystyle{plainnat}
\bibliography{references}

\end{document}